\documentclass[conference, letterpaper]{IEEEtran}

 \usepackage{amsmath,amssymb}
 \usepackage{subfigure}
 \usepackage{graphicx,graphics,color,psfrag}
 \usepackage{cite,balance}
 \usepackage{caption}
 \allowdisplaybreaks
 \usepackage{algorithm}
 \usepackage{accents}
 \usepackage{amsthm}
 \usepackage{bm}
 \usepackage{algorithmic}
 \usepackage[english]{babel}
 \usepackage{multirow}
 \usepackage{enumerate}
 \usepackage{cases}
 \usepackage{stfloats}
 \usepackage{dsfont}
 \usepackage{color,soul}
 \usepackage{amsfonts}
 \usepackage{cite,graphicx,amsmath,amssymb}
 \usepackage{subfigure}
 \usepackage{fancyhdr}
 \usepackage{hhline}
 \usepackage{graphicx,graphics}
 \usepackage{array,color}
 \usepackage{amsmath}

\newtheorem{theorem}{Theorem}

\newtheorem{lemma}{Lemma}

\newtheorem{remark}{\bf Remark}

\def\E{\mathsf{E}}

\def\phi{\varphi}

\def\l{\left}
\def\r{\right}
\def\({\left(}
\def\){\right)}

\def\b0{{\mathbf{0}}}

\newcommand{\nn}{\nonumber}

\begin{document}

\addtolength{\textfloatsep}{-2pt}
\textheight=24cm

\title{Wirelessly Powered Backscatter Communication Networks: Modeling, Coverage and Capacity}

\author{\IEEEauthorblockN{Kaifeng Han and Kaibin Huang}
\IEEEauthorblockA{Department of Electrical and Electronic Engineering\\
The University of Hong Kong, Hong Kong\\
Email: kfhan@eee.hku.hk, huangkb@eee.hku.hk}}\vspace{-10pt}

\maketitle

\begin{abstract} Future Internet-of-Things (IoT) will connect billions of small computing devices embedded in the environment and support their device-to-device (D2D) communication. Powering this massive number of embedded  devices is a key challenge of designing IoT since batteries increase the devices' form factors and their recharging/replacement is difficult. To tackle this challenge, we propose a novel network architecture that integrates wireless power transfer and backscatter communication, called \emph{wirelessly powered backscatter communication} (WP-BC)  networks. In this architecture, \emph{power beacons} (PBs)  are deployed for wirelessly powering devices; their    ad-hoc communication  relies on backscattering and modulating incident continuous waves  from PBs, which consumes orders-of-magnitude less power than traditional radios. Thereby, the dense deployment of low-complexity PBs with high transmission power can  power a large-scale IoT. In this paper, a WP-BC network  is modeled as a random Poisson cluster process in the horizontal plane where PBs are Poisson distributed and active ad-hoc pairs of backscatter communication  nodes with fixed separation distances form random clusters centered at PBs. Furthermore, by harvesting energy from and backscattering radio frequency (RF) waves transmitted by PBs, the transmission power of each node  depends on the distance from the associated PB. Applying stochastic geometry, the network coverage probability and  transmission capacity are derived and optimized as functions of the  backscatter duty cycle and reflection coefficient  as well as the PB density. The effects of the parameters on  network performance are characterized.

\end{abstract}

\section{Introduction}

The vision of \emph{Internet-of-Things} (IoT) is to connect billions of small computing devices embedded in the environment (e.g., walls and furniture) and implanted in bodies and enable their \emph{device-to-device} (D2D) wireless communication. Powering  a massive number of such devices is a key design challenge for IoT.  Batteries add to their weights and form factors and battery recharging/replacement increases  the maintenance cost if not infeasible.  To tackle the challenge, we propose a novel network architecture that enables large-scale passive IoT deployment by seamless integration of wireless power transfer (WPT) \cite{Huang:CuttingLastWiress:2014, bi:PowerCommunication:2014} and low-power backscattering communication,  called a \emph{wirelessly powered backscatter communication} (WP-BC) network. Specifically, \emph{power beacons} (PBs) that are  stations  dedicated for WPT \cite{HuangLauArXiv:EnablingWPTinCellularNetworks:2013} are deployed for wirelessly  powering dense  backscatter D2D  links  and each node transmits data by reflecting and modulating  the carrier signal sent by PBs.
In this paper, a large-scale WP-BC network  is modeled as Poisson cluster processes and its coverage and capacity are analyzed using stochastic geometry.


Backscatter communication refers to a design where a radio device transmits data via reflecting and modulating an incident radio frequency (RF) signal by adapting the level of antenna impedance mismatch to vary the reflection coefficient and furthermore harvests energy from the signal \cite{42stockman1984communication, 1boyer2014backscatterComm.}. As their requires no energy hungry components such as oscillators and analog-to-digital converters (ADCs), a backscatter transmitter consumes power orders-of-magnitude less than a conventional radio. Traditionally, backscatter communication is widely used  in the application of radio frequency Identification (RFID) where a reader powers and communicates with a RFID tag over a short range typically of several meters \cite{1boyer2014backscatterComm., G.Yang:backscattercommun.:2015, Wang:2012:EfficientReliable}. This design is unsuitable for IoT since typical nodes are energy constrained and may not be able to wirelessly power other nodes for communications over sufficiently  long distances. This motivated the design of backscatter communication powered by RF energy harvesting where the transmission of a backscatter node relies on reflecting incident RF signals from the ambient environment such as TV, Wi-Fi and cellular  signals \cite{45liu2013ambient, Kellogg:2014:WIFIBackscatter, Liu2014:Enabling}.   Nevertheless, backscatter communication networks based on ambient RF energy harvesting do not have scalability due to their dependance on other  networks as energy sources. Thus they  may not be suitable for implementing large-scale dense IoT. This motivates the design of WP-BC network  architecture where WPT can deliver power much higher than that by energy harvesting and low-complexity backhaul-less PBs allow wide-spread deployment to power dense passive D2D links.

The work is based on the popular approach of designing and analyzing wireless networks using stochastic geometry (a survey can be found in e.g., \cite{HaenggiAndrews:StochasticGeometryRandomGraphWirelessNetworks}). Among various types of spatial point processes, \emph{Poisson cluster process} (PCP), where \emph{daughter} points form random clusters centered at points from a \emph{parent} \emph{Poisson point process} (PPP),  are commonly used for modeling wireless networks with random cluster topologies arising from geographical factors or protocols for medium access control \cite{GantiHaenggi:OutageClusteredMANET:2009, gulati2010statistics}. In particular, in recent work on heterogeneous networks, PCPs have been frequently used to model the phenomenons  of user clustering at hotspots \cite{chun2015modeling} and the clustering of small-cell base stations (BSs) around macro-cell BSs \cite{suryaprakash2015modeling}. In this work, the  WP-BC network  is also modeled as a PCP where PBs form the parent PPP and backscatter nodes are the clustered daughter points. This topology is motivated by the fact that only nodes sufficiently near PBs can harvest sufficient energy for operating circuits and powering transmission.  Relying on WPT from PBs, nodes' transmission powers depend on their distances from the nearest PBs. In contrast, in the conventional network models,  transmission powers of BSs/nodes are independent of their locations. The location-dependent transmission powers in the WP-BC network   as well as other practical factors (e.g., circuit power and backscatter duty cycle) introduce new challenges for network performance analysis.

Recently, stochastic geometry has been also applied to model  large-scale WPT networks building on existing network architectures including cellular networks \cite{HuangLauArXiv:EnablingWPTinCellularNetworks:2013, che2015spatial} and relay networks \cite{krikidis2014simultaneous, mekikis2014wireless}.  In particular, the WP-BC network  is  similar to cellular networks with WPT considered in \cite{HuangLauArXiv:EnablingWPTinCellularNetworks:2013, che2015spatial} in that  PBs are deployed to power passive nodes'  transmissions. Nevertheless, the current work faces new theoretical challenges arising from a new network topology based on a PCP instead of PPPs in the prior work. Furthermore, practical factors arising from backscatter (e.g., backscatter duty cycle and reflection coefficient)  also introduce a new dimension for network performance optimization.

To the best of our knowledge, the current work represents the first attempt to model and analyze a large-scale backscatter communication network using stochastic geometry.  The theoretic  contributions of this paper are summarized as follows. Based on the mentioned model,  the performance of the WP-BC network are quantified in terms of 1) \emph{success probability} for communication over a typical  backscatter D2D link and 2) \emph{transmission capacity} measuring the spatial density of reliable active  links. The analysis of the metrics are based on deriving the interference characteristic functionals and signal power distribution in the WP-BC network, which account for circuit power, backscatter duty cycle $D$, and reflection coefficient $\beta$ of backscatter nodes. Both success probability and network capacity are found to be concave functions of $D$ and $\beta$, shedding light on the WP-BC network design by convex optimization.

\section{Mathematical Models and Metrics}

\subsection{Network  Model}
The random WP-BC network  is modeled using a PCP as follows. Let $\Pi = \{Y_0, Y_1, \cdots\}$ denote a PPP in the horizontal plane with density $\lambda_p$ modeling the locations of PBs. Consider a cluster of mobile transmitting nodes centered at the origin, denoted as $\widetilde{\mathcal{N}}= \{X_0, X_1, \cdots, X_N\}$. The number of nodes, $N$, is a Poisson random variable (r.v.) with mean $\bar{c}$. The r.v.,   $X_n\in \mathds{R}^2$,  represents the location of the corresponding  node  and $\{X_n\}$ are independent and identically distributed (i.i.d.). For an arbitrary r.v. $X_n$, the direction is isotropic and the distance to the origin, $|X_n|$ \footnote{Given $X\in\mathds{R}^2$, $|X|$ denotes the Euclidean distance from $X$ to the origin.}, has one of two possible probability density functions (PDFs), resulting in the Matern and Thomas cluster process,  as given below:
\begin{align}
\text{(Matern c.p.)}   \quad f(x) &=
   \begin{cases}
       \frac{1}{\pi a^{2}}, &0 \leq |x| \leq a, \\
     0, &\mbox{otherwise,}
   \end{cases}\label{Eq:PDF:Matern}\\
\text{(Thomas c.p.)}\quad     f(x)&=\frac{1}{2\pi \sigma^{2}}\exp\left ( -\frac{|x|^{2}}{2\sigma^{2}} \right), \label{Eq:PDF:Thomas}
\end{align}
where $a$ and $\sigma^2$ are positive constants representing the cluster radius and the variance, respectively.   Let $\{\widetilde{\mathcal{N}}_Y\}$ denote a sequence of clusters constructed by generating an  i.i.d. sequence of clusters having the same distribution as $\widetilde{\mathcal{N}}$ and translating them to be centered at the points $\{Y\}\in \Pi$.
Then the process of transmitting nodes, denoted as $\tilde{\Phi}$, can be written as
\begin{equation}
\tilde{\Phi} = \bigcup_{Y\in \Pi} (\widetilde{\mathcal{N}}_Y + Y).
\end{equation}
The density of $\tilde{\Phi}$ is $\lambda_p \bar{c}$.  Fig.\ref{generalPB} shows  two network realizations  generated based on the Matern and Thomas cluster process. Each transmitting node is paired with an intended receiving node that is located at a unit distance and in an  isotropic direction. This generates a random  spatial process  modeling  distributed D2D links.


\begin{figure}[t]
\centering
\subfigure[Matern cluster process]{\includegraphics[width=8.5cm]{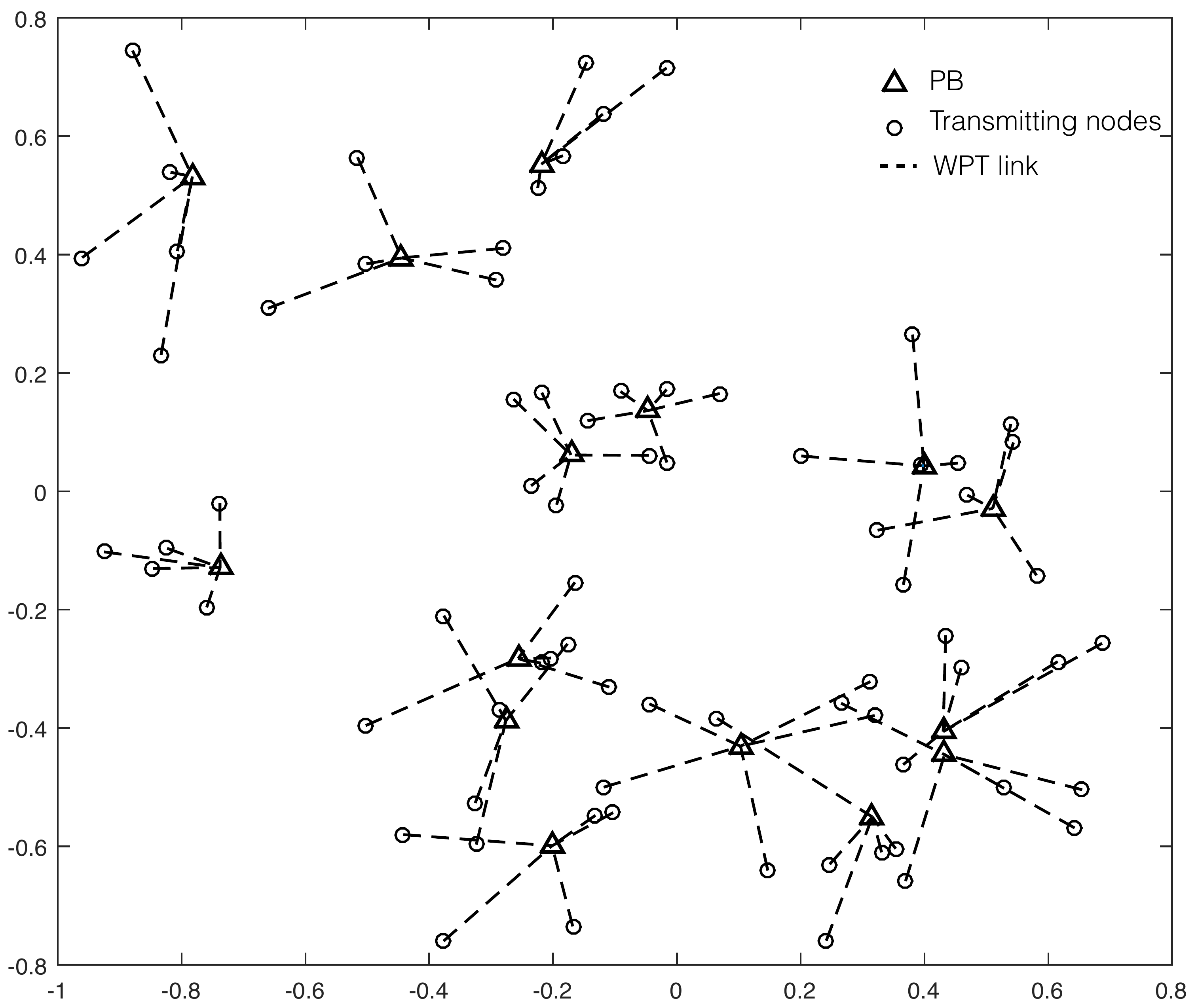}}
\subfigure[Thomas cluster process]{\includegraphics[width=8.5cm]{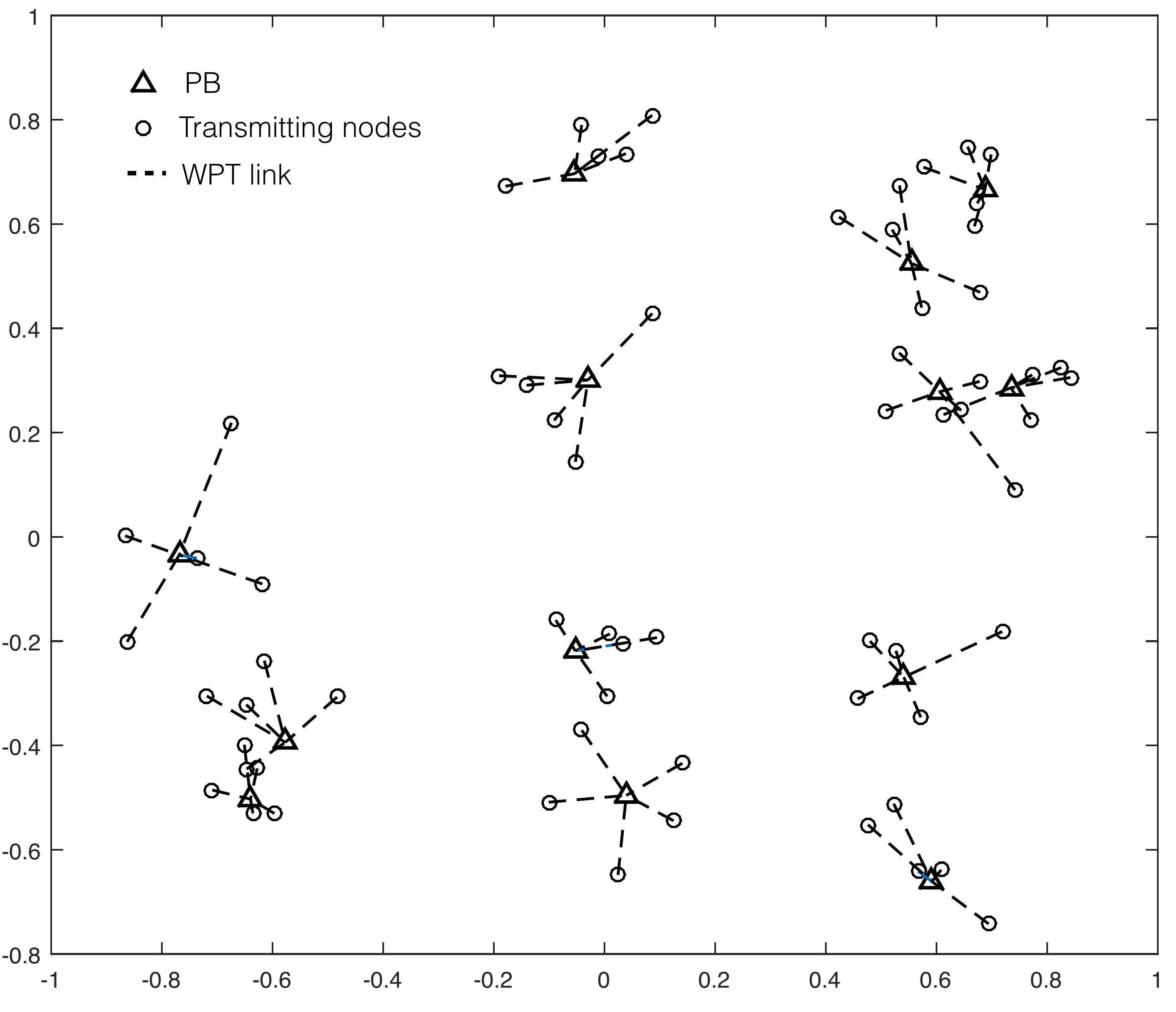}}
\caption{The spatial distribution of the WP-BC network modeled using the (a) Matern cluster process and (b) Thomas cluster process.}\label{generalPB}
\end{figure}

Time is divided into slots of unit duration.  Each slot is further divided into $M$ mini-slots. In each slot, independent of others, a transmitting node randomly selects a single mini-slot to transmit signal by   backscattering.  This divides each slot into a backscatter phase and a waiting    phase of durations $1/M$ and $(1 - 1/M)$, respectively (see more details in the following sub-section). The duty cycle,  denoted as $D$, is given as $D = 1/M$. A transmitting node in a backscatter phase is called a \emph{backscatter node}. Then the backscatter-node process, denoted as $\Phi$, and a cluster of  backscatter nodes centered at $Y$, denoted as $\mathcal{N}_Y$, can be obtained from $\tilde{\Phi}$ and $\widetilde{\mathcal{N}}_Y$ by independent thinning. As a result, $\Phi$ has the density of $\lambda_p \bar{c}D$ and the expected number of nodes in $\mathcal{N}_Y$ is $\bar{c} D$.

The channels are modeled as follows. PBs are equipped with antenna arrays and nodes have single isotropic antennas. Each PB beams a continuous wave (CW) to nodes in the corresponding cluster.  Given beamforming and relatively short distances for efficient WPT, each WPT link can be suitably modeled as a channel with path loss but no fading \cite{Huang:CuttingLastWiress:2014}. The PB allocates transmission power of $\eta$ for each node. As a result, with  a typical PB at $Y_0$, the receive power at a typical  node $X_0$ is given as $P_{X_0} = \eta g |X_0 - Y_0|^{-\alpha_1}$ where $g>0$ denotes the beamforming gain and  $\alpha_1$ the path-loss exponent for WPT links. Due to beamforming, it is assumed that each node harvests negligible energy from other PBs and data signals compared with that from the serving PB.  When transmitting, a  node  backscatters a fraction, called a \emph{reflection coefficient} and denoted as $\beta \in [0, 1]$, of $P_X$ such that the signal power received  at the typical receiver at $Z_0$ is $\beta P_{X_0} h_{X_0} $ where $h_{X_0}\sim\exp(1)$ models Rayleigh fading.  A backscatter node may not be able to transmit if there is insufficient energy for operating its circuit as discussed in the sequel.  Let $Q_X$ denote the random on/off transmission power of the backscatter node $X$. The interference power measured at $Z_0$ can be written as
\begin{equation}
I = \sum_{X \in \Phi\backslash \{X_0\}} Q_X h_X |X - Z_0|^{-\alpha_2},\label{Eq:IntPwr}
\end{equation}
where $\{h_X\}$ are i.i.d. $\exp(1)$ r.v.s modeling Rayleigh fading  and $\alpha_2$ represents the path-loss exponent for interference (D2D) links.

\begin{figure}[t]
\centering
\includegraphics[width=9cm]{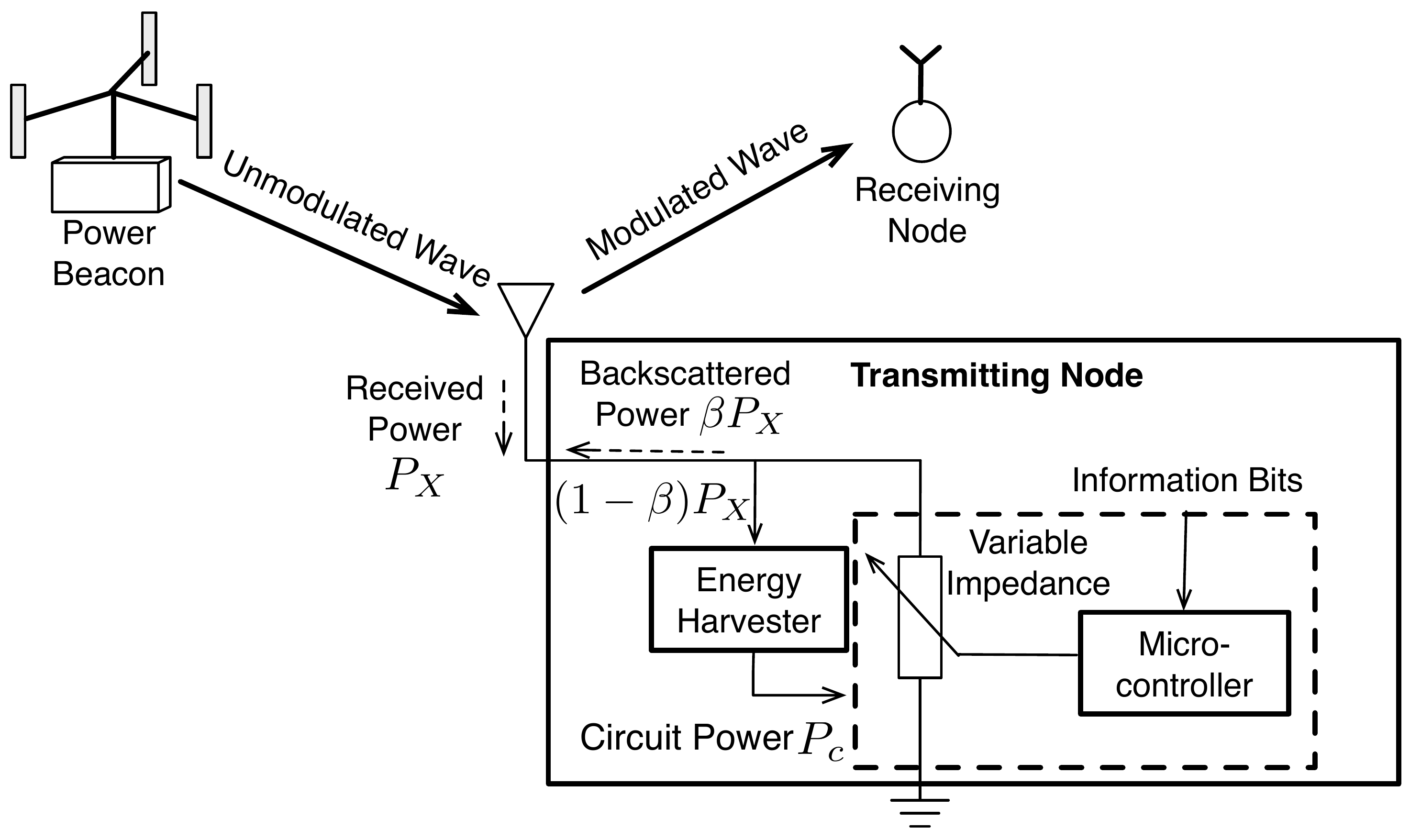}
\caption{Wirelessly powered backscatter communication.}\label{sysmodel}
\end{figure}

\subsection{Backscatter Communication Model}

The operation of WP-BC network is illustrated in Fig.~\ref{sysmodel}. Consider the backscatter phase of an arbitrary slot. A transmitting node adapts the variable impedance (or equivalently its level of mismatch with the antenna impedance)  shown in the figure so as to modulate the backscattered CW with information bits \cite{1boyer2014backscatterComm.}. Given a backscatter node at $X$ and the reflection coefficient $\beta$, the  backscattered power is $\beta P_X$ with the remainder $(1-\beta) P_X$ consumed by the circuit or harvested \cite{karthaus2003:fully}. Next,  for the waiting  phase,  the transmitting node withholds transmission and  performs only energy harvesting. It is assumed that the circuit of each transmitting node consumes fixed power denoted as $P_c$. To be able to transmit, a backscatter  node has to harvest sufficient energy for powering the circuit, resulting in the following \emph{circuit-power constraint}: $(1-\beta) P_X D + P_X (1 - D) \geq P_c$. This gives
\begin{equation}\label{Eq:Circuit}
\text{(Circuit-power constraint)}\quad P_X \geq \frac{P_c}{1- \beta D}.
\end{equation}
Consequently, a backscatter node transmits or is silent depending on if the constraint is satisfied.

\subsection{Performance Metrics}

The network performance is measured by two metrics. One is the probability of the event that the transmission over a typical D2D link is successful, called the \emph{success probability} and denoted as $P_s$. Assuming an interference limited network, the condition for successful transmission is that the receive signal-to-interference ratio (SIR) exceeds a fixed positive threshold $\theta$. Under the circuit power constraint in \eqref{Eq:Circuit}, a transmission power of the typical transmitting node can be written as  $Q_{X_0} = \beta \ell(P_{X_0}) $ where the function $\ell(P)$ gives $P$ if $P \geq P_c/(1- \beta D)$ or else is equal to $0$. Similarly, the interference power in \eqref{Eq:IntPwr} can be rewritten as
\begin{equation}
I  = \sum_{X \in \Phi\backslash\{X_0\}} \beta \ell(P_X) h_X |X - Z_0|^{-\alpha_2}.\label{Eq:IntPwr:a}
\end{equation}
Then the success probability is given as
\begin{equation}\label{Eq:Ps}
P_s = \Pr\l(\beta \ell(P_{X_0})h_{X_0}   \geq \theta I\r).
\end{equation}
The other metric is \emph{transmission capacity} \cite{HaenggiAndrews:StochasticGeometryRandomGraphWirelessNetworks} denoted as $C$ and defined as:
\begin{equation}\label{Eq:TXCap}
C = \lambda_p \bar{c} D P_s.
\end{equation}
The metric measures the density of reliable and active backscatter D2D links.

\section{Interference and Signal  Distributions}
In this section, for the WP-BC network, the distributions of interference at the typical receiver and the signal transmission power for the typical backscatter node are analyzed.  The results are used subsequently for characterizing network coverage and capacity.

\subsection{Interference Characteristic Functionals}
Let $\mathcal{C}(s)$ with $s > 0$ denote the characteristic functional of the interference power $I$ given in \eqref{Eq:IntPwr:a}:  $\mathcal{C}(s) = \E\l[e^{-sI}\r]$.  In this section, the characteristic functional is derived. Without loss of generality, consider a typical backscatter  node at $X_0$ at the origin and the typical receiving node $Z_0 = z$. To facilitate derivation, $I$ is decomposed into the power of \emph{intra-cluster} and \emph{inter-cluster} interference, denoted as $I_a$ and $I_b$, respectively. Mathematically, $I = I_a + I_b$  where
\begin{align}
I_a &= \sum_{X\in \mathcal{N}_{0} \backslash \{ X_0\}} \beta \ell(P_X)  h_X |X - z |^{-\alpha_2}, \label{Eq:IntraInt}\\
I_b &= \sum_{Y\in \Pi\backslash \{Y_0\}  }\sum_{X\in \mathcal{N}_{Y} } \beta \ell(P_X)  h_X |X-z|^{-\alpha_2}. \label{Eq:InterInt}
\end{align}
Note that in $\eqref{Eq:InterInt}$, the first summation is over all other PBs not affiliated with the typical backscatter node (corresponding to clusters of interferers) and the second summation is over the cluster of interferers  centered at the PB $Y$.

The characteristic functionals of $I_a$ and $I_b$ are denoted as $\mathcal{C}_a(s)$ and $\mathcal{C}_b(s)$, respectively, which are defined similarly as $\mathcal{C}(s)$.  They are derived as shown in the following two lemmas.

\begin{lemma}[Intra-cluster interference]\label{Lem:IntraInt}\emph{Given $s \geq 0$, the characteristic functional of the intra-cluster interference power $I_a$ is given as
\begin{equation}
\mathcal{C}_a(s) = \int_{\mathds{R}^2}\exp\Big(-\bar{c} D q(s, y, z) \Big)f(y) d y , \nn
\end{equation}
where
\begin{align}
q(s, y, z) = \int_{\mathds{O}}\frac{1}{1 + s^{-1}\beta^{-1}|x|^{\alpha_{1}}|x - y -z|^{\alpha_{2}}}f(x)d x \nn
\end{align}
and the set $\mathds{O}$ arising from the \emph{circuit power constraint} is defined as $\mathds{O} = \l\{x \in \mathds{R}^2\mid |x| \leq \l(\frac{\eta g(1-\beta D)}{P_c}\r)^{\frac{1}{\alpha_1}}\r\}$.
}
\end{lemma}
\begin{proof} Let $\hat{\E}$ denote the expectation conditioned on the typical backscatter/receiving nodes.

\begin{align}
\hat{\E}\left[e^{-s I_a}\right]  &\overset{(a)}{=} \hat{\E}\Bigg[\! \exp\Bigg(\!\!\!-s\!\!  \sum\limits_{X\in \mathcal{N}_{0} }\!\! \beta h_X \ell(|X\!-\!Y_0|^{-\alpha_{1}}\!) |X\!-\!z|^{-\alpha_{2}}\!\!\Bigg)\! \Bigg] \nonumber \\
& \overset{(b)}{=} \hat{\E}_{Y_0}\Bigg[ \exp\Bigg(-\bar{c} D \int\limits_{\mathds{R}^2}\Big( 1 - \E_h\Big(\exp\Big(- s \beta h \times \nn\\
& \qquad  \quad \ell(|x-Y_0|^{-\alpha_{1}})|x-z|^{-\alpha_{2}}\Big)\Big) f(x - Y_0)d x\Bigg)\Bigg]\nn\\
& = \hat{\E}_{Y_0}\Bigg[ \exp\Bigg(\!\!-\bar{c} D \int\limits_{\mathds{R}^2}\frac{1}{1 + \frac{1}{s \beta \ell(|x-Y_0|^{-\alpha_{1}})|x-z|^{-\alpha_{2}}}}\times\nn\\
& \qquad \qquad   f(x - Y_0)d x\Bigg)\Bigg]\nn\\
&=\hat{\E}_{Y_0}\Bigg[ \exp\Bigg(\!\!\!-\bar{c} D \int\limits_{\mathds{R}^2}\frac{1}{1 + \frac{1}{s \beta \ell(|x|^{-\alpha_{1}})|x - Y_0 -z|^{-\alpha_{2}}}}\times \nn\\
& \qquad \qquad f(x)d x\Bigg)\Bigg], \nn
\end{align}
where $(a)$ and $(b)$ apply   Slivnyak's Theorem an Campbell's Theorem, respectively.  The desired result is obtained using the conditional distribution of $Y_0$ and the definition of $\ell(\cdot)$ and $\E\left[e^{-s I}\right] = \hat{\E}\left[e^{-s I}\right]$ based on Slivnyak's Theorem. \end{proof}

\begin{lemma}[Inter-cluster interference]\label{Lem:InterInt}\emph{Given $s \geq 0$, the characteristic functional of the inter-cluster interference power $I_b$ is given as
\begin{equation}
\mathcal{C}_b(s) =\exp\Bigg(-\lambda_p \int_{\mathds{R}^2} \Big(1 - e^{-\bar{c} D q(s, y, z)}  \Big)dy\Bigg),
\end{equation}
where $q(s, y, z)$  and the set $\mathds{O}$ are defined in Lemma~\ref{Lem:IntraInt}.
}
\end{lemma}

\begin{proof} Using the definition of $I_b$ in \eqref{Eq:InterInt} and applying Slivnyak's Theorem,
\begin{align}
\E\left[e^{-s I_b}\right] &= \E\Bigg[ \exp\Bigg(-s \sum_{Y\in \Pi }\sum_{X\in \mathcal{N}_{Y} } \beta h_X\ell(|X - Y|^{-\alpha_1})  \times \nn\\
&\qquad \qquad  |X-z|^{-\alpha_2} \Bigg) \Bigg] \nonumber \\
&=\E\Bigg[ \prod_{Y\in \Pi}\E\Bigg[\prod_{X\in \mathcal{N}_{Y} }\exp\Big(-s  \beta  h_X \ell(|X - Y|^{-\alpha_1})\times \nn\\
& \qquad \qquad  |X-z|^{-\alpha_2} \Big) \Bigg]\Bigg] \nonumber.
\end{align}
The inner expectation focusing on a single cluster  can be derived using similar steps as in the proof of Lemma~\ref{Lem:IntraInt}. As a result,
\begin{align}
\E\left[e^{-s I_b}\right] &= \E\Bigg[ \prod_{Y\in \Pi} \exp\Big(-\bar{c} D q(s, Y, z) \Big)\Bigg].
\end{align}
Applying Campbell's Theorem gives the desired result.  \end{proof}

\subsection{Signal Distribution}

Under the circuit-power  constraint, there exists a threshold on the separation distance between a pair of PB and affiliated backscatter  node:
\begin{equation}
d_0 = \l[\frac{\eta g (1 - \beta D)}{P_c}\r]^{\frac{1}{\alpha_1}}\label{Eq:Th:Dist}
\end{equation}
such that  the node's transmission power is zero if the distance exceeds the threshold. Then transmission power of the typical backscatter node, denoted as $P_t$, is given as $P_t =  \beta \eta g |X_0 - Y_0|^{-\alpha_1}$ if $|X_0 - Y_0| \leq d_0$ or otherwise $P_t = 0$. The event of $P_t = 0$ corresponds that of (circuit) power outage.
It follows that the \emph{power-outage probability}, denoted as $p_0$, can be written as
\begin{equation}
p_0 = \Pr(P_t = 0) = 2\pi \int_{d_0}^\infty f(r) r dr.   \label{Eq:SigDist:1}
\end{equation}
For the case where the circuit-power constraint is satisfied,
\begin{equation}
\Pr(P_t \geq  \tau) = 2\pi \int\limits_0^{(\beta \eta g/\tau)^{\frac{1}{\alpha_1}}} f(r) rdr, \ \  \tau \geq  \frac{\beta P_c}{1 - \beta D}. \label{Eq:SigDist:2}
\end{equation}
Substituting the CDFs in \eqref{Eq:PDF:Matern} and \eqref{Eq:PDF:Thomas} into \eqref{Eq:SigDist:1} and \eqref{Eq:SigDist:2} gives the following result.

\begin{lemma}[Node transmission power] \label{Lem:SigDist}\emph{The transmission power of a typical backscatter node has support of $\{0\}\cup [\frac{\beta P_c}{1-\beta D}, \infty]$.  The power-outage probability, $p_0$, and the CCDF, denoted as $\bar{F}_t$,  are given as follows.
\begin{itemize}
\item[--] ({\bf Matern cluster process})
\begin{align}
p_0  &= \l\{
\begin{aligned}
& 1 -\l(\frac{d_0}{a}\r)^2, && d_0 <  a, \\
& 0, && \text{otherwise.}
\end{aligned}\r.\nn
\end{align}
\begin{align}
\bar{F}_t(\tau) = \Pr(P_t \geq \tau ) &= \l\{
\begin{aligned}
& \frac{1}{a^2}\l(\frac{\beta \eta g}{\tau}\r)^{\frac{2}{\alpha_1}}, && \tau >   \frac{\beta \eta g }{a ^{\alpha_1}}\\
& 1, && \text{otherwise}
\end{aligned}\r.\nn
\end{align}
with  $\tau \in [\frac{\beta P_c}{1-\beta D}, \infty]$.
\item[--] ({\bf Thomas cluster process})
\begin{align}
p_0 &= \exp\l(-\frac{d_0^2}{2\sigma^2}\r), \nn\\
\bar{F}_t(\tau) &= 1 - \exp\l( - \frac{1}{2\sigma^2}\l(\frac{\beta\eta g}{\tau}\r)^{\frac{2}{\alpha_1}}\r) \nn
\end{align}
with  $\tau \in [\frac{\beta P_c}{1-\beta D}, \infty]$.
\end{itemize}
}
\end{lemma}
A sanity check is as follows. The distance threshold $d_0$ in \eqref{Eq:Th:Dist} is a monotone decreasing  function of $\beta D$ and a monotone increasing function of $P_c$. The reason is that  increasing the duty cycle and reflection coefficient leads to higher energy consumption but increasing the circuit power has the opposite effect.  Consequently, the power-outage probability decreases with increasing $d_0$ for both cases in Lemma~\ref{Lem:SigDist}. Next, the CCDFs in Lemma~\ref{Lem:SigDist} are observed to be independent of $D$ but increase with growing $\beta$. The reason is that conditioned on the node transmitting, the transmission power depends only on the incident power from the PB scaled by  $\beta$ but is independent of the duty cycle.

\section{Network Coverage and Capacity}
In this section, the coverage and capacity of the WP-BC network are characterized using the results derived in the preceding section.

\subsection{Network Coverage}
The network coverage is quantified by deriving the success probability, $P_s$ defined in \eqref{Eq:Ps}, as follows. The event of successful transmission by the typical backscatter node  occurs under two conditions: 1) the circuit-power constraint in \eqref{Eq:Circuit} is satisified and 2) under this condition,  the receive SIR exceeds the threshold $\theta$. Therefore, $P_s$ can be written as
\begin{equation}\label{Eq:Ps:Decomp}
P_s =  \Pr\l(P_t h_{X_0}   \geq \theta I\mid P_t \neq 0 \r)\Pr(P_t \neq 0).
\end{equation}
Replacing the transmission power with its minimum value gives a lower bound on $P_s$ as follows:
\begin{align}
P_s &\geq \Pr\l(\frac{\beta P_c h_{X_0}}{1-\beta D}  > \theta I \r) \Pr(P_t \neq 0) \nn\\
&= \E\l[\exp\l(-\frac{\theta I (1 -\beta D)}{\beta P_c}\r)\r](1 - p_0).  \nn
\end{align}
Then the main result of the section follows by substituting the results derived in the preceding section.
\begin{theorem}[Network coverage]\label{Theo:Coverage}\emph{The \emph{success probability} is bounded as
\begin{equation}
P_s \geq (1-p_0) \mathcal{C}\l(\frac{\theta (1 -\beta D)}{\beta P_c}\r),
\end{equation}
where $\mathcal{C}\l(s\r)= \mathcal{C}_a\l(s\r) \mathcal{C}_b\l(s\r)$ is the interference characteristic functional with $\mathcal{C}_a$ and $\mathcal{C}_b$ given in Lemma~\ref{Lem:IntraInt} and Lemma~\ref{Lem:InterInt}, respectively, and $p_0$ is the power outage probability specified in Lemma~\ref{Lem:SigDist}.
}
\end{theorem}

\begin{remark}[Effects of $p_0$] \emph{The success probability is observed to increase \emph{linearly} with the \emph{transmission  probability} of a backscatter node, $(1 - p_0)$, which agrees with intuition.}
\end{remark}

\begin{remark}[Effects of $D$ and $\beta$ on network coverage] \label{Re:Coverage}\emph{The success probability $P_s$ can be maximized over the duty cycle $D$ and the reflection coefficient $\beta$. A too large or a too small values for each parameter both have a negative effect on network coverage (or the success probability).  A large duty cycle  can result in dense interferers and hence  strong interference but its being too small results in  long waiting period for each node, both reduce $P_s$. Consider $\beta$. On one hand,  increasing $\beta$ scales up transmission power for each node, which can lead to strong interference. On the other hand, $\beta$ being too small leads to  weak receive signal. Both decrease the  success probability.}
\end{remark}

\subsection{Network Capacity} In this section, we consider a WP-BC network with close-to-full network coverage such that transmitted data is always successfully received almost surely. Using \eqref{Eq:Ps:Decomp}, the successful probability can be approximated as $P_s \approx  1 - p_0$. Accordingly, the transmission capacity defined in \eqref{Eq:TXCap} reduces to the density of transmitting nodes:
\begin{equation}
C \approx \lambda_p \bar{c} D (1-p_0) .
\end{equation}
Substituting the results in Lemma~\ref{Lem:SigDist} gives Theorem~\ref{Theo:TxCap}.

\begin{theorem}[Network capacity]\label{Theo:TxCap}\emph{In the regime of  close-to-full network coverage, the network transmission capacity can be approximated  as follows.
\begin{itemize}
\item[--] ({\bf Matern cluster process})
\begin{equation}
C \approx  \frac{\lambda_p \bar{c} D}{a^2}\l[\frac{\eta g (1 - \beta D) }{P_c}\r]^{\frac{2}{\alpha_1}}. \nn
\end{equation}
\item[--] ({\bf Thomas cluster process})
\begin{equation}
C \approx  \lambda_p \bar{c} D \l[1 - \exp\l( -\frac{1}{\sigma^2}\l(\frac{\eta g (1-\beta D)  }{P_c}\r)^{\frac{2}{\alpha_1}}\r)\r]. \nn
\end{equation}
\end{itemize}
}
\end{theorem}

First of all, the transmission  capacity $C$ is observed to be proportional to the density of backscatter nodes that is consistent with intuition.

\begin{remark}[Effects of $D$ and $\beta$ on network capacity] \emph{The parameters affect the transmission capacity in the mentioned regime  by  varying the transmitting-node density. In comparison, their effects on network coverage are not entirely the same and  reflected in those on transmission probability and link reliability (see Remark~\ref{Re:Coverage}). In the  regime of almost-full network coverage,  $C$ decreases  with the growing reflection coefficient $\beta$. The reason is that  a large coefficient leads to less harvested energy and thereby reduces the transmitting-node density. In particular, $C$ scales with $\beta$  as $(1 - D\beta)^{\frac{2}{\alpha_1}}$. Next, increasing $D$ has two opposite effects on the transmitting-node density, namely increasing the backscatter-node density but reducing transmission probability due to less harvested energy. Therefore, the capacity can be optimized over $D$. For instance, for the model based on the Matern cluster process, the maximum capacity is
\begin{equation}
\max_D C(D) = \frac{\lambda_p \bar{c}\alpha_1}{a^2(2 + \alpha_1 \beta)}\l[\frac{2 \eta g }{P_c (2 + \alpha_1\beta)}\r]^{\frac{2}{\alpha_1}}
\end{equation}
and the optimal duty cycle is given as $D^* = \min\l(1, \frac{\alpha_1}{2+\alpha_1 \beta}\r)$. This assumes that $D^*$  is within the constrained range  discussed in the following remark.  The capacity optimization for the case of Thomas cluster process is similar but more tedious.
 }
\end{remark}

\begin{remark}[Constraints on $D$ and $\beta$] \label{Re:Cap}\emph{It is clear from Remark~\ref{Re:Coverage} that the consideration of the mentioned network operational regime constraints $D$ and $\beta$ to certain ranges to ensure link reliability. The capacity results in Theorem~\ref{Theo:TxCap} holds only for the parameters falling in these ranges.  The corresponding region for $(\beta, D)$ can be derived by bounding the conditional  probability in \eqref{Eq:Ps:Decomp} by a positive value close to one. For instance, using Theorem~$1$, an inner bound of  the region can be derived as
\begin{equation}
\l\{(D, \beta) \in [0, 1]^2  \mid \mathcal{C}\l(\frac{\theta (1 -\beta D)}{\beta P_c}\r) \geq 1 - \epsilon \r\}
\end{equation}
where the positive constant $\epsilon \approx 0$.}
\end{remark}

\section{Simulation Results}

The parameters for the simulation are set as follows unless stated otherwise. The PB transmission power $\eta$ = 40 dBm (10 W) and circuit power is  $P_{c}$ = 7 dBm. The SIR threshold is set as $\theta = -5$ dB in the typical range for ensuring almost-full network coverage (see e.g., \cite{28andrews2011tractable}). The path-loss exponents for WPT and communication links are  $\alpha_{1}$ = 3 and $\alpha_{2}$ = 3, respectively. The backscatter reflection coefficient is  $\beta = 0.6$ and duty cycle $D= 0.4$. The PB density  is $\lambda_{p} = 0.2\  /\textrm{m}^{2}$ and the expected  number of nodes in each cluster $\bar{c} = 3$. The transmission distances for D2D links are set as $1$ m. The network model based on the Thomas cluster process is assumed with the parameter $\sigma^{2}$ = 4.

\begin{figure}[t]
\centering
\subfigure[Effect of the reflection coefficient]{\includegraphics[width=8.5cm]{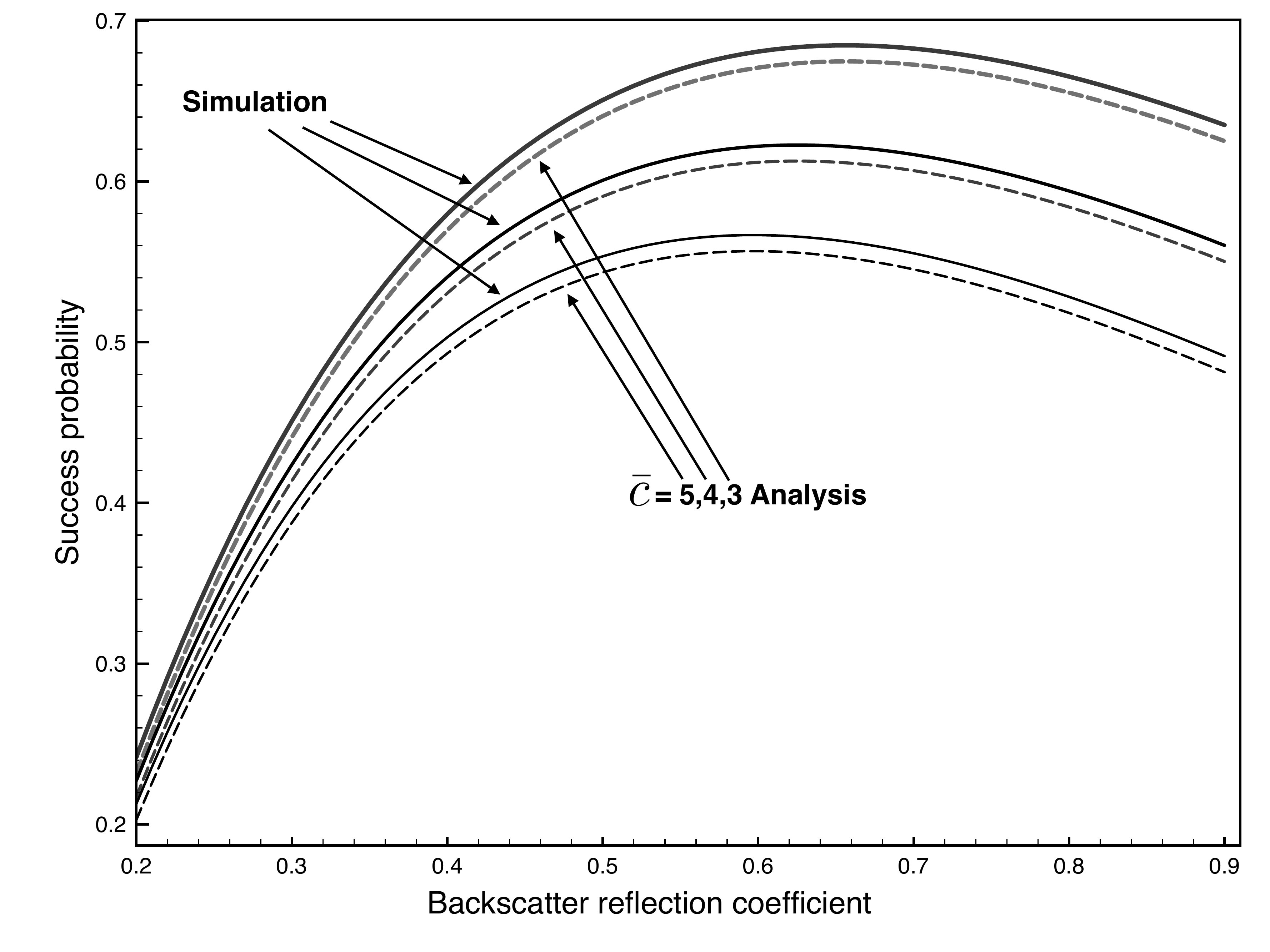}}
\subfigure[Effect of the duty cycle]{\includegraphics[width=8.5cm]{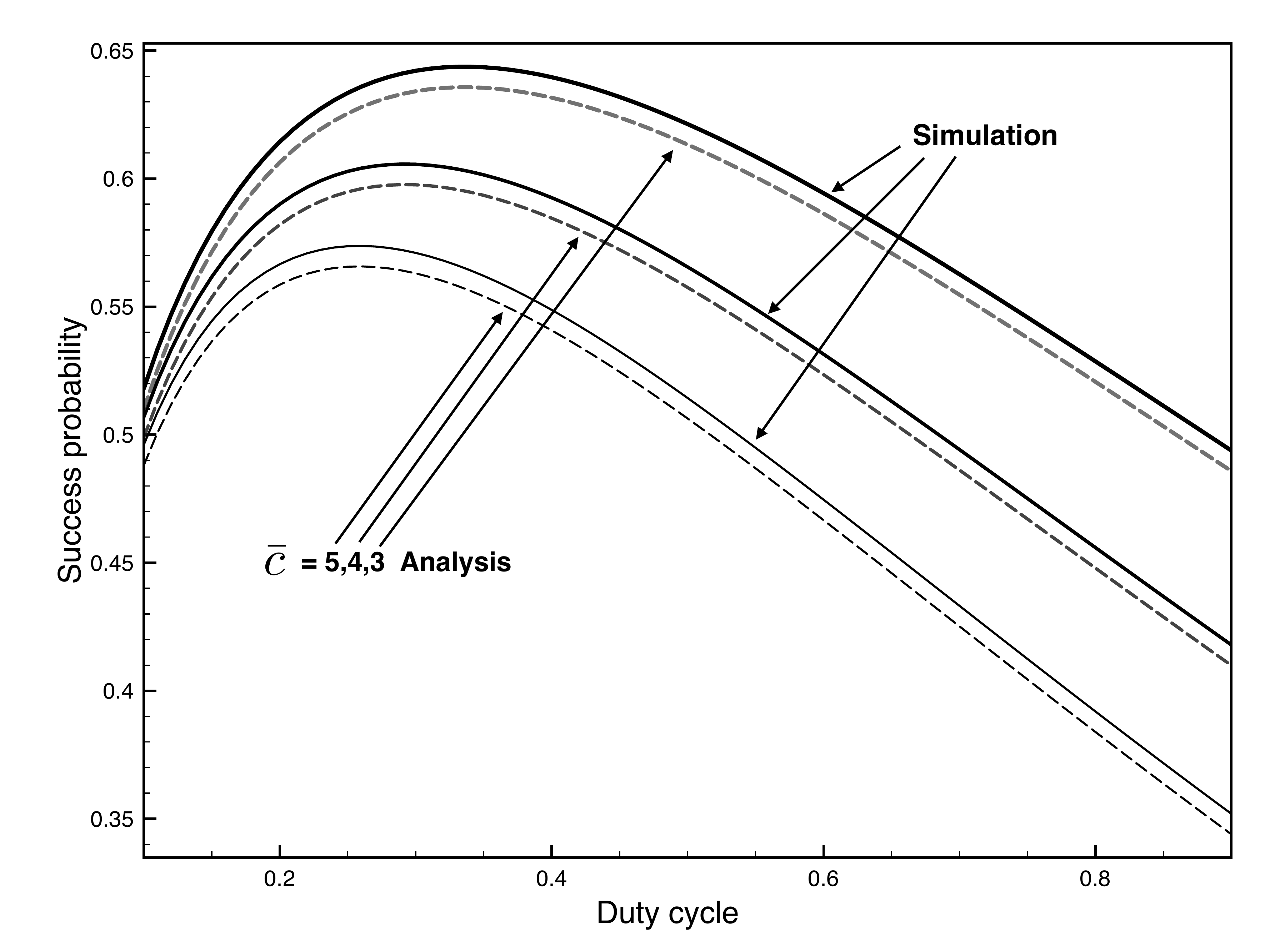}}
\caption{The effects of the backscatter parameters, namely the duty cycle and (backscatter) reflection coefficient,  on the success probability for a variable expected number of backscatter nodes per cluster,  $\bar{c} \in \{3,4,5\}$. }\label{cpSim}
\end{figure}

The curves of success probability versus the backscatter duty cycle and reflection coefficient are plotted in  Fig.~\ref{cpSim} for different values of $\bar{c}$. The curves based on the analytical results in Theorem~\ref{Theo:Coverage} are plotted for comparison. It is observed that the theoretical lower  bounds are tight. The curves show that the success probability is concave  functions of the backscatter parameters, which is consistent with the discussion in Remark~\ref{Re:Coverage}. The  optimal values for the reflection coefficient and duty cycle are observed to be about $0.6$ and $0.3 - 0.35$, respectively.

The curves of network transmission capacity  versus the PB density and the reflection coefficient are shown in Fig.~\ref{Fig:TxCap} for different values of $\bar{c}$. When the density of PB is relatively small, the network capacity is observed to grow linearly with the PB density. For a large PB density, the capacity saturates as the network becomes dense and interference limited. Next, the network capacity is observed to be a concave function of the reflection coefficient. In the region with $\beta \geq 0.6$ corresponding to high network coverage [see Fig.\ref{cpSim}(a)], the capacity decreases with growing $\beta$ that is consistent with Remark~\ref{Re:Cap}.

\section{Conclusion}
In this work, we have proposed the new network architecture, namely the WP-BC network,  for realizing dense backscatter communication networks using wireless power transfer enabled by PBs. A large-scale WP-BC network has been modeled using the PCP. Applying  stochastic geometry theory, the success probability and the transmission capacity have  been derived to quantify the performance of network coverage and capacity, respectively. In particular, the results  relate the network performance to  the backscatter parameters, namely the duty cycle and the reflection coefficient.

\begin{figure*}[t]
\centering
\subfigure[Effect of the PB density ]{\includegraphics[width=8.5cm]{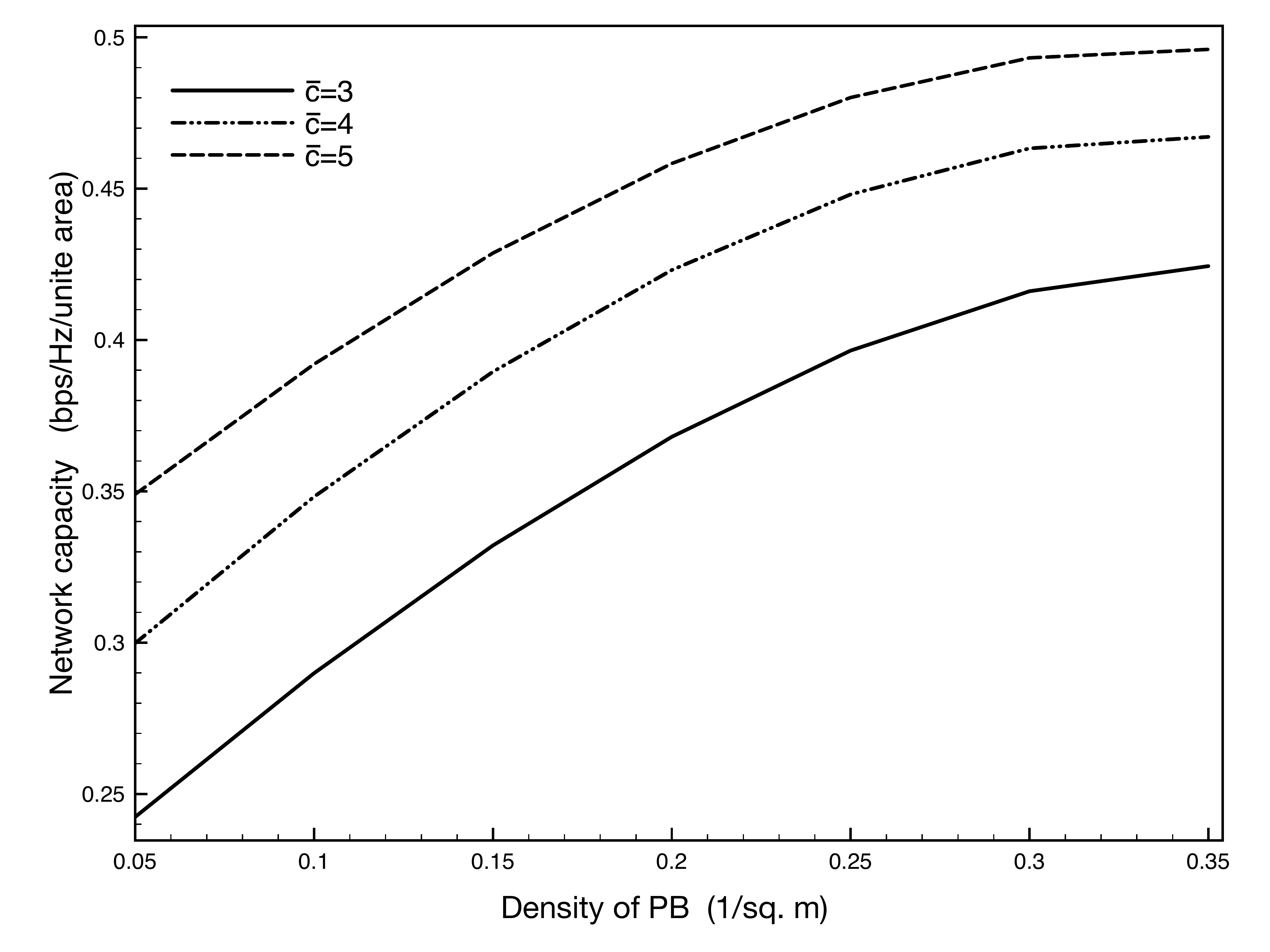}}
\subfigure[Effect of the reflection coefficient]{\includegraphics[width=8.5cm]{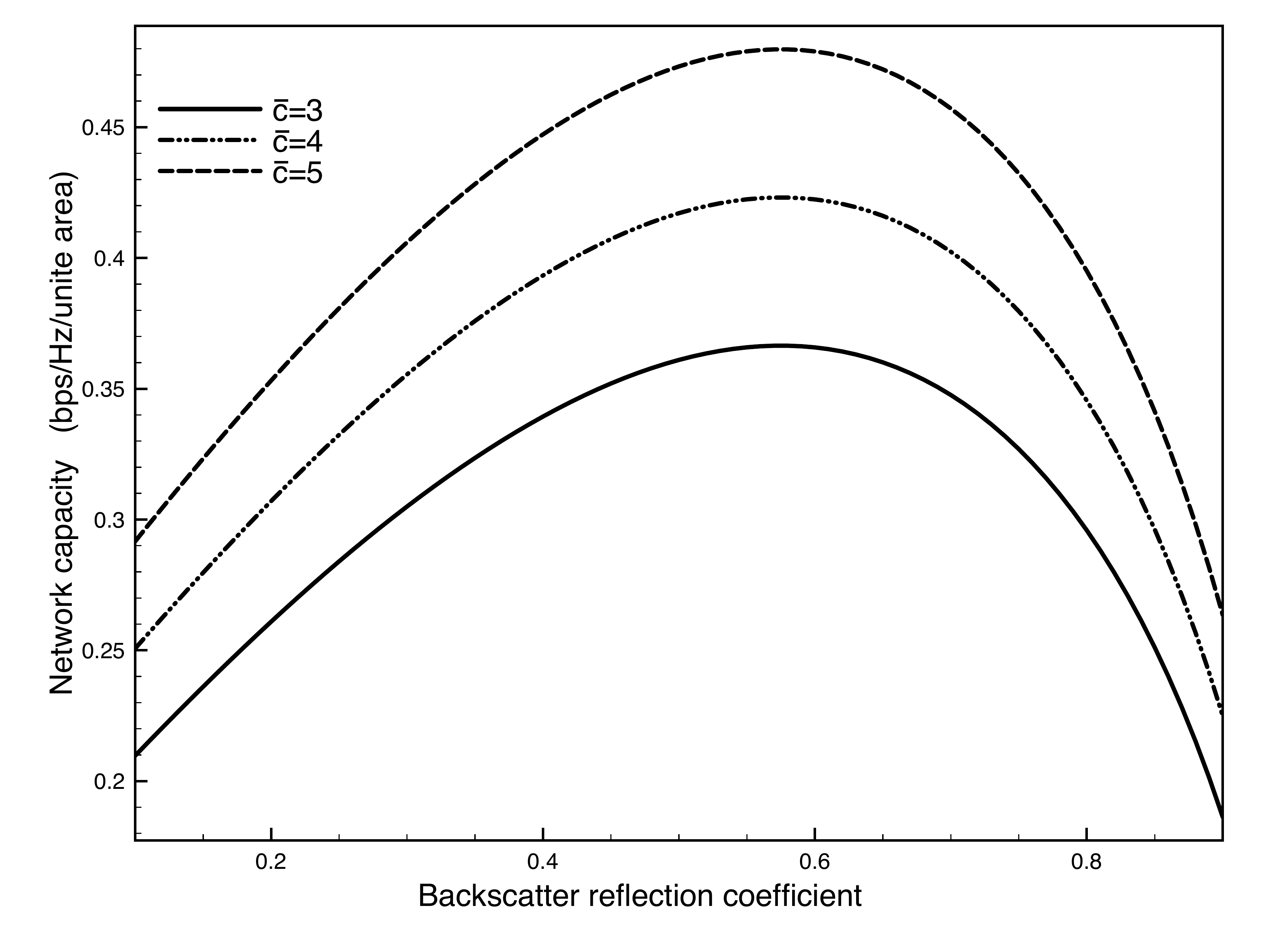}}
\caption{The effects of the PB density and reflection coefficient on the network capacity for a variable expected number of backscatter nodes per cluster,  $\bar{c} \in \{3,4,5\}$}
\label{Fig:TxCap}
\end{figure*}
%

\end{document}